\documentclass[preprint,5p,12pt,twocolumn]{elsarticle}

\usepackage{amssymb,amsmath,colonequals,xspace,tikz}
\usetikzlibrary{calc,decorations.pathmorphing,fit,backgrounds,positioning,arrows,shapes}
\tikzstyle{knoten}=[circle,draw=black,thin,fill=white,inner sep=0pt,minimum size=4.5mm]
\tikzstyle{knotenklein}=[circle,draw=black,thin,fill=white,inner sep=0pt,minimum size=2.5mm]

\DeclareMathOperator{\MCF}{MCF}

\newtheorem{theorem}{Theorem}
\newtheorem{lemma}{Lemma}
\newtheorem{corollary}{Corollary}
\newproof{proof}{Proof}
\newdefinition{definition}{Definition}

\renewcommand{\o}[1]{\overline{#1}}
\newcommand{\CMCFPC}{$\textsc{BCMCFP}_{\mathbb{R}}$\xspace}

\makeatletter
\def\NAT@spacechar{~}% NEW
\makeatother

\makeatletter
\def\ps@pprintTitle{%
 \let\@oddhead\@empty
 \let\@evenhead\@empty
 \def\@oddfoot{}%
 \let\@evenfoot\@oddfoot}
\makeatother

\begin{document}

\begin{frontmatter}

\title{On the Complexity and Approximability of Budget-Constrained Minimum Cost Flows\tnoteref{titleref}}

\tnotetext[titleref]{This work was partially supported by the German Federal Ministry of Education and Research within the project ``SinOptiKom -- Cross-sectoral Optimization of Transformation Processes in Municipal Infrastructures in Rural Areas''.}

\author[TUKL]{Michael Holzhauser\corref{cor1}}
\ead{holzhauser@mathematik.uni-kl.de}

\author[TUKL]{Sven O. Krumke}
\ead{krumke@mathematik.uni-kl.de}

\author[TUKL]{Clemens Thielen}
\ead{thielen@mathematik.uni-kl.de}

\cortext[cor1]{Corresponding author. Fax: +49 (631) 205-4737. Phone: +49 (631) 205-2511}

\address[TUKL]{University of Kaiserslautern, Department of Mathematics\\
  Paul-Ehrlich-Str.~14, D-67663~Kaiserslautern, Germany\\ \quad}

\begin{abstract}
	We investigate the complexity and approximability of the \emph{budget-constrained minimum cost flow problem}, which is an extension of the traditional minimum cost flow problem by a second kind of costs associated with each edge, whose total value in a feasible flow is constrained by a given budget~$B$. This problem can, e.g., be seen as the application of the $\varepsilon$-constraint method to the bicriteria minimum cost flow problem. We show that we can solve the problem exactly in weakly polynomial time~$\mathcal{O}(\log M \cdot \MCF(m,n,C,U))$, where $C$, $U$, and $M$ are upper bounds on the largest absolute cost, largest capacity, and largest absolute value of any number occuring in the input, respectively, and $\MCF(m,n,C,U)$ denotes the complexity of finding a traditional minimum cost flow. Moreover, we present two fully polynomial-time approximation schemes for the problem on general graphs and one with an improved running-time for the problem on acyclic graphs.
\end{abstract}

\begin{keyword}
	algorithms \sep complexity \sep minimum cost flow \sep approximation
\end{keyword}

\end{frontmatter}

\section{Introduction}
\label{sec:Intro}

In this paper, we investigate the natural extension of the traditional minimum cost flow problem (cf., e.g., \citep{Ahuja}) by a second kind of costs, called \emph{usage fees}, which are linear in the flow on the corresponding edge and bounded by a given budget~$B$. This extension allows us to solve many related problems such as the budget-constrained maximum dynamic flow problem (since each dynamic flow can be represented as a traditional minimum cost flow (cf. \citep{FordFulkersonDynamicFlow})) or the application of the $\varepsilon$-constraint method to the bicriteria minimum cost flow problem (cf., e.g., \citep{ChankongHaimes}).

To the best of our knowledge, the budget-constrained minimum cost flow problem was first mentioned in \citep{Ahuja}, where a structural result but no combinatorial algorithm was presented. A related problem in which a \emph{fixed} usage fee is induced by each edge with positive flow was investigated by \citet{SpieksmaAccessibility}. The model that we use here was recently investigated in \citet{BudgetConstrainedMinCostFlows}, where a strongly polynomial-time algorithm based on the interpretation of the problem as a bicriteria minimum cost flow problem was derived.

We extend these results and show that, using similar ideas, we can also obtain a weakly polynomial-time combinatorial algorithm that performs worse only by a logarithmic factor than the best algorithm for the traditional minimum cost flow problem. Moreover, we present two fully polynomial-time approximation schemes (FPTAS), one of which is based on techniques introduced by \citet{PapadimitrouApproximatePareto} and has a weakly polynomial running-time and one of which is based on the packing-LP framework developed by \citet{GargKoenemann}, achieving a strongly-polynomial running-time. The running-time of the latter FPTAS is subsequently improved for the case of acyclic graphs.

\section{Preliminaries}
\label{sec:Problem}

\subsection{Problem Definition}

In the budget-constrained minimum cost flow problem (abbreviated as \emph{\CMCFPC} in the following), we are given a directed multigraph~$G=(V,E)$ with edge capacities~$u_e \in \mathbb{N}_{\geq 0}$, costs~$c_e \in \mathbb{Z}$ (i.e., we allow integral costs with arbitrary sign), and usage fees~$b_e \in \mathbb{N}_{\geq 0}$ per unit of flow on the edges~$e \in E$, as well as a budget~$B \in \mathbb{N}_{\geq 0}$ and a distinguished source~$s \in V$ and sink~$t \in V$. The aim is to find a feasible $s$-$t$-flow~$x$ in $G$ that minimizes~$\sum_{e \in E} c_e \cdot x_e$ subject to the budget-constraint~$\sum_{e \in E} b_e \cdot x_e \leq B$. The problem~\CMCFPC can be stated as a linear program as follows:
{\allowdisplaybreaks
\begin{subequations}\label{eqn:CMCFP:LP}
\begin{align}
	\min		&	\sum_{e \in E} c_e \cdot x_e \label{eqn:CMCFP:LP_Objective} \\
	\text{s.t.}	&	\sum_{e \in \delta^-(v)} x_e - \sum_{e \in \delta^+(v)} x_e = 0 && \forall v \in V \setminus \{s,t\}, \\
                &	\sum_{e \in E} b_e \cdot x_e \leq B,	\label{eqn:CMCFP:LP_Budget} \\
				&	0 \leq x_e \leq u_e && \forall e \in E.
\end{align}
\end{subequations}
}
Here, we denote by $\delta^+(v)$ ($\delta^-(v)$) the set of \emph{outgoing} (\emph{incoming}) edges of some node~$v \in V$. We assume that there are no nodes~$v \in V \setminus \{s,t\}$ with $\delta^+(v) = \emptyset$ or $\delta^-(v) = \emptyset$ since no flow can reach such nodes due to flow conservation. Note that we can detect and remove such nodes along with their incident edges in linear time.

Furthermore, note that since the zero-flow is always feasible and has objective value zero, the optimal objective value of each instance of \CMCFPC is always non-positive. The problem is a generalization of the problem variant in which a desired flow value~$F$ is given since we can ``enforce'' such a flow value by adding an edge with negative costs of large absolute value and capacity~$F$ (cf. \citep{BudgetConstrainedMinCostFlows} for further details).

\subsection{Approximation Algorithms}

An algorithm~$A$ is called a (polynomial-time) \emph{approximation algorithm with performance guarantee $\alpha \in [1,\infty)$} or simply an \emph{$\alpha$--approximation} for \CMCFPC if, for each instance~$I$ of \CMCFPC with optimum solution~$x^*$, it computes a feasible solution~$x$ with objective value $c(x) \leq \frac{1}{\alpha} c(x^*)$ in polynomial time (note that $c(x)$ and $c(x^*)$ are non-positive, so $\frac{1}{\alpha} c(x^*) \geq c(x^*)$). An algorithm~$A$ that receives as input an instance~$I\in\Pi$ and a real number~$\varepsilon \in (0,1)$ is called a \emph{polynomial-time approximation scheme (PTAS)} if, on input~$(I,\varepsilon)$, it computes a feasible solution~$x$ with objective value $c(x) \leq (1 - \varepsilon) \cdot c(x^*)$ with a running-time that is polynomial in the encoding size~$|I|$ of $I$. If this running-time is additionally polynomial in $\frac{1}{\varepsilon}$, the algorithm is called a \emph{fully polynomial-time approximation scheme (FPTAS)}.

Moreover, for the case of \CMCFPC, we call an algorithm a \emph{bicriteria FPTAS} if, for each $\varepsilon \in (0,1)$, it computes a solution~$x$ with $c(x) \leq (1 - \varepsilon) \cdot c(x^*)$ and $b(x) \leq (1 + \varepsilon) \cdot b(x^*)$ in polynomial time.

\subsection{Parametric Search}

Throughout the paper, we make use of Megiddo's parametric search technique (cf. \citep{MegiddoCombinatorialOptimization}), which can be described as follows: Assume that we want to solve an optimization problem~$\Pi$ for which we already know an (exact) algorithm~$A$ that solves the problem, but in which some of the input values are now \emph{linear parametric values} that depend linearly on some real parameter~$\lambda$. Moreover, suppose that an algorithm~$C$ is known (in the following called \emph{callback}) that is able to decide if some candidate value for $\lambda$ is smaller, larger, or equal to the value~$\lambda^*$ that leads to an optimum solution to the underlying problem~$\Pi$. The idea of the parametric search technique is to simulate the execution of algorithm~$A$ with variables that still depend on the symbolic value~$\lambda$, and to continue the execution until we reach a comparison of two linear parametric values that needs to be resolved. Since both values depend linearly on $\lambda$, it either holds that one of the variables is always larger than or equal to the other one (in which case the result of the comparison is independent from $\lambda$) or that there is a unique intersection point~$\lambda'$. For this intersection point, we evaluate the callback~$C$ in order to determine if $\lambda' < \lambda^*$, $\lambda' > \lambda^*$, or $\lambda' = \lambda^*$ and, thus, resolve the comparison and continue the execution. Hence, as soon as the simulation of $A$ finishes, we have obtained an optimum solution to $\Pi$. The overall running-time is given by the running-time of $A$ times the running-time of $C$ and can be further improved using parallelization techniques described in \citep{MegiddoParallel}. We refer to \citep{MegiddoCombinatorialOptimization} for further details on the parametric search technique. Further applications and extenions of parametric search techniques can moreover be found in \citep{CohenFixedDimension,ToledoApproximateParametricSearching,ToledoFixedDimension}.

\section{Exact Algorithms}
\label{sec:Exact}

We start with results on the complexity of the problem \CMCFPC. The mathematical model \eqref{eqn:CMCFP:LP_Objective} -- \eqref{eqn:CMCFP:LP_Budget} for \CMCFPC, as introduced in Section~\ref{sec:Problem}, is a linear program, which can be solved in weakly polynomial time by known techniques such as interior point methods (cf. \citep{Schrijver}). In particular, using the procedure described by \citet{VaidyaInteriorPoint} to our multigraph setting, we get the following weakly polynomial running-time for \CMCFPC:

\begin{theorem}
	\CMCFPC is solvable in weakly polynomial time~$\mathcal{O}(m^{2.5} \cdot \log M)$. \qed
\end{theorem}

However, in this paper, we are interested in \emph{combinatorial} algorithms that exploit the structure of the underlying problem. We show how we can incorporate combinatorial algorithms for the traditional minimum cost flow problem in order to solve the more general budget-constrained minimum cost flow problem \CMCFPC.

We can solve \CMCFPC by computing an efficient solution of the bicriteria minimum cost flow problem with the two objective functions~$c(x)$ and $b(x)$ that minimizes $c(x)$ while maintaining $b(x) \leq B$. Graphically, each optimum solution~$x^*$ of \CMCFPC corresponds to a point in the objective space that lies on the pareto frontier and not above the line~$b = B$.\footnotemark\xspace The situation is shown in Figure~\ref{fig:CMCFP:CMCFPC_pareto}.
\footnotetext{We refer to \citep{Ehrgott} for an in-depth treatment of bicriteria optimization problems and efficient solutions.}

\begin{figure}[ht]
	\centering
	\begin{tikzpicture}
		\begin{scope}
			\clip (0,0) -- (-2.5,1) -- (-4,2.5) -- (-5,5) -- (-5.2,6) -- (0,6) -| (0,0);
			\fill[gray!30] (0,0) rectangle (-5.2,6);
		\end{scope}
		\draw[thick] (0,0) -- (-2.5,1) -- (-4,2.5) -- (-5,5) -- (-5.2,6);

		\draw[dotted] (-5,3.5) -- (0.5,3.5);
		\draw (1,3.5) node {$b=B$};

		\draw (-4.4,3.5) circle (1mm) node[anchor=south west] {\scriptsize $(c(x^*),b(x^*))^T$};

		\draw[->] (-5,0) -- (1,0) node[anchor=north] {$c$};
		\draw[->] (0,0) node[anchor=north] {$0$} -- (0,6) node[anchor=west] {$b$};
	\end{tikzpicture}
	\caption{The objective space of the interpretation of \CMCFPC as a bicriteria minimum cost flow problem. The gray area corresponds to the set of the objective values of feasible flows, the thick black lines correspond to efficient edges, which form the pareto frontier.}
	\label{fig:CMCFP:CMCFPC_pareto}
\end{figure}
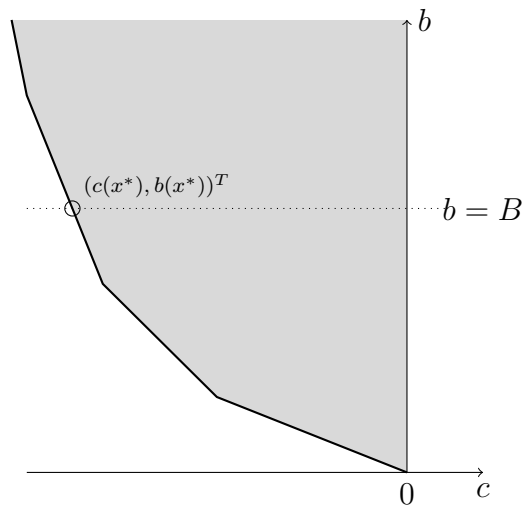

It is well-known that, for each point~$(c,b)^T$ on the pareto frontier, there is some value~$\lambda \in [0,\infty)$ and a feasible flow~$x$ with $(c(x),b(x))^T = (c,b)^T$ such that $x$ is a minimum cost flow with respect to the costs~$b_e + \lambda \cdot c_e$ for each edge~$e \in E$ (cf. \citet{Geoffrion}). Assume that there are two flows~$x^1$ and $x^2$ that are both optimal for some specific value of $\lambda$, i.e., $b(x^1) + \lambda \cdot c(x^1) = b(x^2) + \lambda \cdot c(x^2) = \alpha$ for some value~$\alpha$. Then, for both of the flows~$x^i$ with $i \in \{1,2\}$, it holds that $b(x^i) = \alpha - \lambda \cdot c(x^i)$, i.e., they lie on the same efficient edge, which is a straight line with slope $-\lambda$ in the objective space. In other words, computing a minimum cost flow with edge-costs~$b_e + \lambda \cdot c_e$ will either provide a solution that corresponds to on an extreme point of the pareto frontier or some point that lies on the efficient edge with slope~$-\lambda$. Moreover, as shown in \citep{BudgetConstrainedMinCostFlows}, the slopes of these efficient edges differ by minimum absolute amounts:

\begin{lemma}[\citep{BudgetConstrainedMinCostFlows}]\label{lem:CMCFP:Pareto_Slopes}
	The slopes of two efficient edges on the pareto frontier of any instance of \CMCFPC differ by an absolute value of at least $\frac{1}{\o{c}^2}$ for $\o{c} \colonequals \sum_{e \in E} |u_e \cdot c_e|$. \qed
\end{lemma}

As explained above, each optimum solution $x^*$ of \CMCFPC is a minimum cost flow with respect to the edge-costs~$c_e + \lambda^* \cdot b_e$ for at least one value~$\lambda^* \in [0,+\infty)$. In particular, if $\Lambda^*$ denotes the set of all these values~$\lambda^*$, it holds that $\Lambda^*$ is a closed interval containing either one or infinitely many such values~$\lambda^*$ depending on whether the optimum solutions correspond to points that lie amid or at the corner of some efficient edge in the objective space, respectively. As claimed in the following lemma, which is proven in \citep{BudgetConstrainedMinCostFlows}, we are able to decide the membership in $\Lambda^*$ efficiently:

\begin{lemma}[\citep{BudgetConstrainedMinCostFlows}]\label{lem:CMCFP:Pareto_Callback}
	Let $\Lambda^* \neq \emptyset$ denote the set of parameters~$\lambda^*$ for which an optimum solution~$x^*$ to \CMCFPC is a minimum cost flow with respect to the edge-costs $c_e + \lambda^* \cdot b_e$ for each $e \in E$. For some candidate value~$\lambda$, it is possible to decide whether $\lambda < \min \Lambda^*$, $\lambda > \max \Lambda^*$, or $\lambda \in \Lambda^*$ in $\mathcal{O}(\MCF(m,n,C,U))$~time. \qed
\end{lemma}

Let $\MCF(m,n)$ denote the complexity of computing a traditional minimum-cost flow in \emph{strongly} polynomial time. As shown in \citep{BudgetConstrainedMinCostFlows}, Lemma~\ref{lem:CMCFP:Pareto_Callback} can be used within Megiddo's parametric search technique in order to obtain strongly polynomial-time algorithms for \CMCFPC:

\begin{theorem}[\citep{BudgetConstrainedMinCostFlows}]
	\CMCFPC is solvable in strongly polynomial time $\mathcal{O}(m \log m \cdot \min\{T_1(m,n),T_2(m,n),T_3(m,n)\})$ with
	\begin{itemize}
		\item $T_1(m,n) \in \mathcal{O}((m + n \log n) \cdot \MCF(m,n))$,
		\item $T_2(m,n) \in \mathcal{O}((n \cdot \log\frac{m}{n} + n \cdot \log n + \log m) \cdot \MCF(m,n) + m)$, and
		\item $T_3(m,n) \in \mathcal{O}(\log \log m \cdot \log^2 m \cdot \MCF(m,n) + f(m))$ for $f(m) \in o(m^3)$. \qed
	\end{itemize}
\end{theorem}

We now show how we can use Lemma~\ref{lem:CMCFP:Pareto_Callback} within a binary search in order to obtain a weakly polynomial-time algorithm that performs within a factor~$\mathcal{O}(\log M)$ of each algorithm for the traditional minimum cost flow problem:

\begin{theorem}\label{thm:CMCFP:CMCFPC_WeaklyP}
	\CMCFPC is solvable in weakly polynomial time $\mathcal{O}(\log M \cdot \MCF(m,n,C,U))$.
\end{theorem}

\begin{proof}
	Consider the set $K \colonequals \left\{ k \cdot \frac{1}{2\o{c}^2} : k \in \{0, \ldots, \o{b} \cdot 2\o{c}^2 \} \right\}$, where $\o{b} \colonequals \sum_{e \in E} u_e \cdot b_e \in \mathcal{O}(m \cdot M^2)$ and $\o{c} \colonequals \sum_{e \in E} | u_e \cdot c_e | \in \mathcal{O}(m \cdot M^2)$ are upper bounds on the total usage fees and total absolute value of the costs of any feasible flow, respectively. Note that each extreme point of the pareto frontier can be obtained by a minimum cost flow computation with edge-costs~$c_e + \lambda \cdot b_e$ for some $\lambda \in K$ since the slopes of any two efficient edges differ by an absolute amount of at least $\frac{1}{\o{c}^2}$ according to Lemma~\ref{lem:CMCFP:Pareto_Slopes}. Hence, by incorporating the procedure that is described in Lemma~\ref{lem:CMCFP:Pareto_Callback} into a binary search on the set~$K$, we either find some value~$\lambda \in \Lambda^*$ (in which case we have also found an optimum solution to \CMCFPC) or two ``adjacent'' values~$\lambda^{(1)} \colonequals k \cdot \frac{1}{2\o{c}^2}$ and $\lambda^{(2)} \colonequals (k+1) \cdot \frac{1}{2\o{c}^2}$ for some $k \in \{0, \ldots, \o{b} \cdot 2\o{c}^2 -1 \}$ with $\lambda^{(1)} < \min \Lambda^*$ and $\lambda^{(2)} > \max \Lambda^*$. These values, however, yield solutions $x^{(1)}$ and $x^{(2)}$ that correspond to the corner points of the same efficient edge, which crosses the line $b = B$ in the objective space. Thus, by computing a suitable convex combination of the two solutions $x^{(1)}$ and $x^{(2)}$, we obtain an optimum solution to \CMCFPC and are done.

	The running-time of the procedure is dominated by the binary search on the set~$K$ and the resulting $\mathcal{O}(\log|K|)$ calls to the procedure that is described in Lemma~\ref{lem:CMCFP:Pareto_Callback}. Hence, the overall running-time is given by
	\begin{align*}
		&\ \mathcal{O}(\log|K| \cdot \MCF(m,n,C,U)) \\
        =&\ \mathcal{O}(\log (\o{b} \cdot \o{c}^2) \cdot \MCF(m,n,C,U)) \\
		=&\ \mathcal{O}(\log (m^3 \cdot M^6) \cdot \MCF(m,n,C,U)) \\
		=&\ \mathcal{O}(\log M \cdot \MCF(m,n,C,U)),
	\end{align*}
	which shows the claim. \qed
\end{proof}

\section{Approximation Algorithms}

\subsection{General Graphs}

In \citep{PapadimitrouApproximatePareto}, the authors show that an \emph{$\varepsilon$-approximate pareto frontier} (i.e., a set of points~$P_\varepsilon$ such that, for each point~$y$ on the pareto frontier~$P$, there is a point~$y' \in P_\varepsilon$ such that $y$ is within a factor of $(1 + \varepsilon)$ from $y'$ in each component) of a linear convex optimization problem with $k$~objective functions can be determined by solving $\mathcal{O}((8L k^2/\varepsilon)^k)$~instances of the problem with only one objective function, where $L$ denotes the encoding-length of the largest possible objective value. Applied to \CMCFPC, we are, thus, able to determine an $\varepsilon$-approximate pareto frontier in $\mathcal{O}\left(\left(\frac{\log M}{\varepsilon} \right)^2 \cdot \MCF(m,n,C,U) \right)$~time, which in turn implies a bicriteria FPTAS for \CMCFPC. The following lemma shows that this also yields a traditional FPTAS for \CMCFPC:

\begin{lemma}
    Any bicriteria FPTAS for \CMCFPC also induces a single-criterion FPTAS for \CMCFPC.
\end{lemma}

\begin{proof}
    For each instance of \CMCFPC with optimum solution~$x^*$, the given bicriteria FPTAS computes a solution $x$ with $c(x) \leq (1 - \varepsilon) \cdot c(x^*)$ and $b(x) \leq (1 + \varepsilon) \cdot b(x^*)$ in time that is polynomial in the instance size and $\frac{1}{\varepsilon}$. Since both the costs~$c$ and usage fees~$b$ are linear functions, it suffices to scale down the given solution as follows: Let $x' \colonequals \frac{x}{1 + \varepsilon}$. Clearly, $x'$ is feasible since it still fulfills every flow conservation and capacity constraint and since $b(x') = \frac{1}{1 + \varepsilon} \cdot b(x) \leq b(x^*) \leq B$. Moreover, for $\varepsilon' \colonequals 2\varepsilon$, it holds that
    \begin{align*}
        c(x') &= \frac{1}{1 + \varepsilon} \cdot c(x) \leq \frac{1 - \varepsilon}{1 + \varepsilon} \cdot c(x^*) \\
        &= \frac{1 - \frac{\varepsilon'}{2}}{1 + \frac{\varepsilon'}{2}} \cdot c(x^*) = \frac{1 - \varepsilon' + \frac{(\varepsilon')^2}{4}}{1 - \frac{(\varepsilon')^2}{4}} \cdot c(x^*) \\
        &\leq (1 - \varepsilon') \cdot c(x^*),
    \end{align*}
    which shows the claim. \qed
\end{proof}

\begin{corollary}
    There is an FPTAS for \CMCFPC that runs in $\mathcal{O}\left(\left(\frac{\log M}{\varepsilon} \right)^2 \cdot \MCF(m,n,C,U) \right)$~time.\qed
\end{corollary}

We now show how we can obtain an FPTAS with \emph{strongly} polynomial running-time using a different approach based on a combination of \citeauthor{GargKoenemann}'s packing-LP framework and Megiddo's parametric search technique. We will therefore need the following auxiliary lemma:

\begin{lemma}\label{lem:MultigraphToSimpleGraph}
    Let $\lambda$ be a parameter with a callback that fulfills $C(m,n) \in \Omega(\frac{m}{\log m})$. Any multigraph~$G$ with linear parametric edge-lengths~$l_e(\lambda)$ on each $e \in E$ can be turned into a simple graph~$G'$ that only contains the shortest edge among all parallel edges between two nodes in $\mathcal{O}(\log m \cdot \log\log m \cdot C(m,n))$~time.
\end{lemma}

\begin{proof}
    Let $S \colonequals \{ (v,w) \in V^2: | \delta^+(v) \cap \delta^-(w) | \geq 2 \}$ denote the set of all pairs of nodes with at least two parallel edges between them. In order to determine the simple graph~$G'$ with the desired properties, we need to evaluate the minimum of all edges in $\delta^+(v) \cap \delta^-(w)$ for each $(v,w) \in S$. As shown in \citep{ValiantParallelMinimum}, we can determine the minimum of $k$ values in $\mathcal{O}(\log\log k)$~time using $\mathcal{O}(k)$~processors. We simulate all of these computations in parallel, which results in a total number of $\mathcal{O}(\sum_{(v,w) \in S} |\delta^+(v) \cap \delta^-(w)|) = \mathcal{O}(m)$~processors. In order to reduce the number of callback calls, we simulate the $\mathcal{O}(m)$~processors sequentially in a round-robin manner until each of them either finishes its computation or holds at the comparison of two linear parametric values, yielding~$\mathcal{O}(m)$ candidate values for $\lambda$ that need to be resolved using the callback for $\lambda$. Using a binary search on the set of these candidate values in combination with a successive determination of the median, which can be employed in $\mathcal{O}(m)$ time according to \citet{MedianLinear}, we can resolve all of the comparisons simultaneously in $\mathcal{O}(\log m \cdot C(m,n) + m)$~time and continue the simulation of the processors. After $\mathcal{O}(\log\log m)$~iterations of the above procedure, each processor has finished its computation and the edge with minimum length is determined for each $(v,w) \in S$, which shows the claim.\qed
\end{proof}

\begin{theorem}\label{thm:FPTAS}
    There is an FPTAS for \CMCFPC that runs in strongly polynomial-time~$\mathcal{\widetilde{O}}\left(\frac{1}{\varepsilon^2} \cdot (m^2 \cdot n + m \cdot n^3)\right)$.\footnotemark
\end{theorem}

\footnotetext{To simplify running-times, it is common to use $\mathcal{\widetilde{O}}(p)$ in order to denote $\mathcal{O}(p \cdot \log^k m)$ with $k \in \mathcal{O}(1)$.}

\begin{proof}
    We consider an equivalent, circulation-based version of \CMCFPC that can be obtained by inserting an edge with infinite capacity, zero costs, and zero usage fees between $t$ and $s$. Then, according to the flow decomposition theorem for traditional flows (cf. \citep{Ahuja}), each optimum flow~$x^*$ is positive on $\mathcal{O}(m)$~simple cycles~$C$ with strictly negative costs~$c(C) \colonequals \sum_{e \in C} c_e < 0$. Let $\mathcal{C}$ denote the set of all such simple cycles with negative costs the underlying graph. For $b(C) \colonequals \sum_{e \in C} b_e$, we obtain the following cycle-based formulation of \CMCFPC:
    {\allowdisplaybreaks
    \begin{align}
        \min~		&	\sum_{C \in \mathcal{C}} c(C) \cdot x_C \nonumber\\
    	\text{s.t.~}&	\sum_{\stackrel{C \in \mathcal{C}:}{e \in C}} x_C \leq u_e && \forall e \in E, \nonumber\\
                    &	\sum_{C \in \mathcal{C}} b(C) \cdot x_C \leq B,	\nonumber\\
    				&	x_C \geq 0 && \forall C \in \mathcal{C}. \nonumber\\
    \intertext{The dual of this linear program can be stated as follows:}
        \min~		&	B \cdot \mu + \sum_{e \in E} u_e \cdot y_e, \nonumber\\
    	\text{s.t.~}&	b(C) \cdot \mu + \sum_{e \in C} y_e \geq -c(C) && \forall C \in \mathcal{C}, \label{eqn:DualConstraint}\\
    				&	y_e \geq 0 && \forall e \in E, \nonumber\\
                    &   \mu \geq 0. \nonumber
    \end{align}%
    }%
    Although the number of variables (constraints) is exponential in the primal (dual), we are able to derive a strongly polynomial-time FPTAS for the problem using the packing-LP framework introduced by \citet{GargKoenemann}. In short, \citeauthor{GargKoenemann} show that, as long as it is possible to determine the \emph{most-violated constraint} of the dual for some infeasible dual solution~$(y,\mu)$ with $y > 0$ and $\mu > 0$ in polynomial time $\mathcal{O}(A(m,n))$, then there is an FPTAS with a running-time in $\mathcal{O}(\frac{1}{\varepsilon^2} \cdot m \log m \cdot A(m,n))$.

    Note that, since $c(C) < 0$ for each $C \in \mathcal{C}$, we can rewrite equation~\eqref{eqn:DualConstraint} as
    \[ \frac{b(C) \cdot \mu + \sum_{e \in C} y_e}{-c(C)} \geq 1, \]
    or, equivalently,
    \[ \frac{\sum_{e \in C} (b_e \cdot \mu + y_e)}{\sum_{e \in C} -c_e} \geq 1. \]
    Hence, we are done if we can determine the \emph{minimum ratio cycle}~$C \in \mathcal{C}$ with edge-costs~$b_e \cdot \mu + y_e$ and edge-times~$-c_e$ in polynomial time. \citet{MegiddoParallel} derived an algorithm that determines the minimum ratio cycle in a \emph{simple} graph in $\mathcal{O}(n^3 \log n + m \cdot n \log^2n \log\log n)$~time by computing all-pair shortest paths in combination with \citeauthor{KarpMinimumMeanCycle}'s minimum mean cycle algorithm \citep{KarpMinimumMeanCycle} as a callback function in his parametric search. In order to comply with our setting of multigraphs, we first need to apply Lemma~\ref{lem:MultigraphToSimpleGraph} with the minimum mean cycle algorithm as a callback, running in $C(m,n) \colonequals \mathcal{O}(m \cdot n)$~time, to the underlying graph, which yields a running-time of $\mathcal{O}(m \cdot n \cdot \log m \cdot \log\log m)$. In total, we get that
    \begin{align*}
        A(m,n) =~   & \mathcal{O}(m \cdot n \cdot \log m \log\log m \\
                    & + n^3 \log n + m \cdot n \cdot \log^2n \log\log n).
    \end{align*}
    Thus, incorporated in \citeauthor{GargKoenemann}'s framework, we obtain an overall running-time of
    \begin{align*}
         &\ \mathcal{O}\left(\frac{1}{\varepsilon^2} \cdot m \log m \cdot A(m,n) \right) \\
         =&\ \mathcal{\widetilde{O}}\left(\frac{1}{\varepsilon^2} \cdot (m^2 \cdot n + m \cdot n^3) \right).
    \end{align*}\qed%
\end{proof}

\subsection{Acyclic Graphs}

We now show how we can improve the running-time of the FPTAS described in Theorem~\ref{thm:FPTAS} for the case of an acyclic graph~$G$. Since there are no cycles in $G$, we only need to repeatedly determine minimum ratio $s$-$t$-paths rather than minimum ratio cycles. This, however, can be done more efficiently, as shown in the following lemma:

\begin{lemma}\label{lem:MostViolatedConstraintAcyclic}
	Let~$d^{(1)} \colon E \rightarrow \mathbb{R}$ and $d^{(2)} \colon E \rightarrow \mathbb{R}$ be two cost functions with $d^{(1)}_e \colonequals d^{(1)}(e)$ and $d^{(2)}_e \colonequals d^{(2)}(e)$ for each $e \in E$ and assume that $\sum_{e \in P} d^{(2)}_e > 0$ for each $s$-$t$-path $P$. An $s$-$t$-path~$P^*$ that minimizes the ratio~$\frac{\sum_{e \in P} d^{(1)}_e}{\sum_{e \in P} d^{(2)}_e}$ among all $s$-$t$-paths~$P$ can be found in $\mathcal{O}(m \cdot (\log m \log\log m + n \cdot \log n))$~time on acyclic graphs.
\end{lemma}

\begin{proof}
	Let $\mathcal{P}$ denote the set of all $s$-$t$-paths in the underlying acyclic graph~$G$. Similar to \citet{MegiddoCombinatorialOptimization} and \citet{MinimalRatioSpanningTrees}, we can restrict our considerations to the problem $\min_{P \in \mathcal{P}} \sum_{e \in P} d^{(\lambda)}_e$ with $d^{(\lambda)}_e \colonequals d^{(1)}_e - \lambda \cdot d^{(2)}_e$ for each $e \in E$ and some parameter~$\lambda$: Using similar arguments as in \citep{MinimalRatioSpanningTrees}, it is easy to see that, for some candidate value of $\lambda$, it holds that $\min_{P \in \mathcal{P}} \sum_{e \in P} d^{(\lambda)}_e$ is negative (positive) if and only if the value of $\lambda$ is smaller (larger) than the value~$\lambda^*$ that leads to an optimum solution~$P^*$ to $\min_{P \in \mathcal{P}} \frac{\sum_{e \in P} d^{(1)}_e}{\sum_{e \in P} d^{(2)}_e}$. Hence, by simulating the shortest path algorithm for acyclic graphs with edge lengths~$d^{(\lambda)}$ for a symbolic value of $\lambda$ using Megiddo's parametric search technique, it is possible to determine the optimum solution~$P^*$ in $\mathcal{O}(m^2)$~time.

    We can improve this running-time by first applying Lemma~\ref{lem:MultigraphToSimpleGraph} to the underlying multigraph in order to obtain a simple graph with the same shortest paths as in $G$ in $\mathcal{O}(m \cdot \log m \log\log m)$~time. In this simple graph, we simulate the shortest path algorithm for acyclic graphs, which initially sets the distance label of each node to infinity. It then investigates the nodes in the order of a topological sorting and, for each outgoing edge~$e=(v,w)$ of some node~$v \in V \setminus \{t\}$ in this sorting, updates the distance label~$dist(w)$ of node~$w$ to $\min\{dist(w), dist(v) + l_e \}$ where $l_e$ denotes the length of edge~$e$, which results in a comparison of two linear parametric values. Note that the edges in $\delta^+(v)$ head to different nodes since the underlying graph is simple, so all of these comparisons are independent from each other. Thus, by evaluating a binary search over the set of candidate values that result from each of these comparisons as described in \citep{MegiddoParallel} and as used above, we only need $\mathcal{O}(\log| \delta^+(v) |)$ shortest path computations at an overhead of $\mathcal{O}(|\delta^+(v)|)$ for the median computations. This results in an running-time for the parametric shortest path computation of
    {\allowdisplaybreaks
	\begin{align*}
		&\ \mathcal{O}\left( \sum_{v \in V \setminus \{t\}} (\log | \delta^+(v) | \cdot m + |\delta^+(v)| ) \right) \\
		=&\ \mathcal{O}\left( m \cdot n \cdot \log n \right),
	\end{align*}%
	}%
    which in, combination with the overhead of $\mathcal{O}(m \cdot \log m \cdot \log\log m)$ for the transformation into a simple graph, shows the claim.
    \qed
\end{proof}

By incorporating the results of Lemma~\ref{lem:MostViolatedConstraintAcyclic} in the packing-LP framework of \citeauthor{GargKoenemann} as described in the proof of Theorem~\ref{thm:FPTAS}, we immediately get the following corollary:

\begin{corollary}
    There is an FPTAS for \CMCFPC that runs in $\mathcal{O}\left(\frac{1}{\varepsilon^2} \cdot m^2 \log m \cdot (\log m \log\log m + n \log n) \right)$ time on acyclic graphs. \qed
\end{corollary}

\section{Conclusion}

In this paper, we presented results on the complexity and approximability of the budget-constrained minimum cost flow problem. As the problem is known to be solvable both in weakly polynomial time by interior-point methods and in strongly-polynomial time as shown in \citep{BudgetConstrainedMinCostFlows}, we developed a new combinatorial algorithm that runs in weakly polynomial time~$\mathcal{O}(\log M \cdot \MCF(m,n,C,U))$. Moreover, we presented a weakly polynomial-time FPTAS that uses the $\varepsilon$-approximate pareto frontier and a strongly polynomial-time FPTAS based on both \citeauthor{GargKoenemann}'s packing-LP framework and Megiddo's parametric search technique. Moreover, we could show that we can improve the running-time of the latter algorithm for the case of acyclic graphs.

% \section*{References}

\bibliographystyle{elsarticle-num-names}
\bibliography{/Users/holzhaus/Documents/Forschung/literature}

\end{document}